\documentclass[11pt,a4paper]{article}
\usepackage[utf8]{inputenc}
\usepackage{amsthm}
\usepackage{amsmath}
\usepackage{amsfonts}
\usepackage{amssymb}
\usepackage{paralist}
\usepackage{graphicx}
\usepackage{a4wide}
\usepackage{comment}
\usepackage{verbatim}
\usepackage[blocks]{authblk}
\usepackage[unicode]{hyperref}
\hypersetup{pdftitle=Online Colored Bin Packing}
\hypersetup{pdfauthor={Martin B\"{o}hm, Ji\v{r}\'{i} Sgall, Pavel Vesel\'{y}}}

\def\O{{\mathcal{O}}}
\let\epsilon=\varepsilon
\def\vari#1{\mathit{#1}}

\DeclareMathOperator*{\argmax}{arg\,max}
\newcommand*{\nolinefrac}[2]{\genfrac{}{}{0pt}{1}{#1}{#2}}
\newcommand*\samethanks[1][\value{footnote}]{\footnotemark[#1]}

\begin{document}

\date{}
\title{Online Colored Bin Packing}
\author[1]{Martin B\"{o}hm\thanks{This work was supported by the
project 14-10003S of GA \v{C}R and by the GAUK project 548214.}\ }
\author[1]{Ji\v{r}\'{i} Sgall\samethanks\ }
\author[1]{Pavel Vesel\'{y}\samethanks\ }
\affil[1]{Computer Science Institute of Charles University, Prague, Czech Republic.
\texttt{\{bohm,sgall,vesely\}@iuuk.mff.cuni.cz}.}

\lefthyphenmin=2
\righthyphenmin=2

\newtheorem*{define}{Definition}
\newtheorem{theorem}{Theorem}[section]
\newtheorem{lemma}[theorem]{Lemma}
\newtheorem{claim}[theorem]{Claim}
\newtheorem{observation}[theorem]{Observation}
\newtheorem{corollary}[theorem]{Corollary}
\newtheorem{proposition}[theorem]{Proposition}

\maketitle

\begin{abstract}
In the Colored Bin Packing problem a sequence of items of sizes up to
$1$ arrives to be packed into bins of unit capacity. Each item has one
of $c\geq 2$ colors and an additional constraint is that we cannot
pack two items of the same color next to each other in the same bin.
The objective is to minimize the number of bins.

In the important special case when all items have size zero, we
characterize the optimal value to be equal to color discrepancy. As
our main result, we give an (asymptotically) 1.5-competitive algorithm
which is optimal. In fact, the algorithm always uses at most
$\lceil1.5\cdot\vari{OPT}\rceil$ bins and we show a matching lower
bound of $\lceil1.5\cdot\vari{OPT}\rceil$ for any value of
$\vari{OPT}\geq 2$.
In particular, the absolute ratio of our algorithm
is $5/3$ and this is optimal.
For items of arbitrary size we give a lower bound of $2.5$
and an absolutely $3.5$-competitive
algorithm.
When the items have sizes of at most $1/d$ for a real $d \geq
2$ the asymptotic competitive ratio is $1.5+d/(d-1)$.
We also show that classical algorithms First Fit, Best Fit and Worst Fit are not
constant competitive, which holds already for three colors and small items.

In the case of two colors---the Black and White Bin Packing
problem---we prove that all Any Fit algorithms have the absolute
competitive ratio $3$. When the items have sizes of at most $1/d$ for a
real $d \geq 2$ we show that the Worst Fit algorithm is absolutely
$(1+d/(d-1))$-competitive.
\end{abstract}

\section{Introduction}

In the \textit{Online Black and White Bin Packing} problem proposed by
Balogh et al.~\cite{balogh13,balogh14} as a generalization of
classical bin packing, we are given a list of items of size in $[0,
  1]$, each item being either black, or white.  The items are coming
one by one and need to be packed into bins of unit capacity.
The items in a bin are ordered by their arrival
time. The additional constraint to capacity is that the
colors inside the bins are alternating, i.e., no two items of the same color
can be next to each other in the same bin.  The goal is to minimize
the number of bins used.

\textit{Online Colored Bin Packing} is a natural generalization of
Black and White Bin Packing in which items can have more than two
colors.  As before, the only additional condition to unit capacity is that we
cannot pack two items of the same color next to each
other in one bin.

Observe that optimal offline packings with and without reordering the
items differ in this model. The packings even differ by a non-constant
factor: Let the input sequence have $n$ black items and then $n$ white
items, all of size zero. The offline optimal number of bins with
reordering is $1$, but an offline packing without reordering (or an
online packing) needs $n$ bins, since the first $n$ black items must
be packed into different bins.  Hence we need to use the offline
optimum without reordering in the analysis of online colored bin
packing algorithms.

There are several well-known and often used algorithms for classical
Bin Packing. We investigate the \textit{Any Fit} family of algorithms (AF). These algorithms
pack an incoming item into some already open bin whenever
it is possible with respect to the size and color constraints.
The choice of the open bin (if more are available) depends on the algorithm.  AF
algorithms thus open a new bin with an incoming item only when there
is no other possibility. Among AF algorithms, \textit{First Fit}
(FF) packs an incoming item into the first bin where it fits (in the
order by creation time), \textit{Best Fit} (BF) chooses the bin with
the highest level where the item fits and \textit{Worst Fit} (WF)
packs the item into the bin with the lowest level where it fits.

\textit{Next Fit} (NF) is more restrictive than Any Fit algorithms,
since it keeps only a single open bin and puts an incoming item into
it whenever the item fits, otherwise the bin is closed and a new one
is opened.

\medskip

\textbf{Previous results.}
Balogh et al.~\cite{balogh13,balogh14} introduced the Black and White
Bin Packing problem.  As the main result, they give an algorithm
{\it Pseudo} with the absolute competitive ratio exactly $3$ in the general
case and $1+d/(d-1)$ in the parametric case, where the items
have sizes of at most $1/d$ for a real $d \geq 2$.  They also proved
that there is no deterministic or randomized online algorithm whose
asymptotic competitiveness is below $1 + \frac{1}{2 \ln 2} \approx
1.721$.

Concerning specific algorithms, they proved that Any Fit algorithms are at most
$5$-competitive and even optimal for zero-size items.  They show input
instances on which FF and BF create asymptotically $3 \cdot
\vari{OPT}$ bins.  For WF there are sequences of items witnessing that it
is at least $3$-competitive and $(1+d/(d-1))$-competitive in
the parametric case for an integer $d \geq 2$.
Furthermore, NF is not constant competitive. 

The idea of the algorithm Pseudo, on which we build as well, is that
it first packs the items regardless of their size, i.e., treating their
size as zero. This can be done optimally for two colors, and the optimum equals the
maximal discrepancy in the sequence of colors (to be defined
below). Then these bins are partitioned by NF into bins of level at
most 1.

In the offline setting, Balogh et al.~\cite{balogh13} gave a $2.5$-approximation
algorithm with $\O(n \log n)$ time complexity and an asymptotic
polynomial time approximation scheme, both when reordering is allowed.

Very recently and independent of us D\'osa and Epstein~\cite{DosEps14}
studied Colored Bin Packing. They improved the lower bound for online
Black and White Bin Packing to $2$ for deterministic algorithms, 
which holds for more colors as well. For at least $3$ colors they
proved an asymptotic lower bound of $1.5$ for zero-size items.
They designed an absolutely 4-competitive algorithm based on Pseudo
and an absolutely $2$-competitive balancing algorithm for
zero-size items. They also showed that BF, FF and WF are not
competitive at all (with non-zero sizes).

\medskip

\textbf{Our results.}
We completely solve the case of Colored Bin Packing for zero-size
items. As we have seen, this case is important for constructing
general algorithms.
The offline optimum (without reordering) is actually not only
lower bounded by the color discrepancy, but equal to it for zero-size
items (see Section~\ref{sec:Prelim}). For online
algorithms, we give an (asymptotically) 1.5-competitive algorithm
which is optimal (see Section~\ref{sec:zeroSizeUB}).
In fact, the algorithm always uses at most
$\lceil1.5\cdot\vari{OPT}\rceil$ bins and we show a matching lower
bound of $\lceil1.5\cdot\vari{OPT}\rceil$ for any value of
$\vari{OPT}\geq 2$ (see Section~\ref{sec:zeroSizeLB}).
This is significantly stronger than the asymptotic
lower bound of $1.5$ of D\'osa and Epstein~\cite{DosEps14}, in
particular it shows that the absolute ratio of our algorithm is $5/3$,
and this is optimal.

For items of arbitrary size and three colors, we show a lower bound of
$2.5$, which breaks the natural barrier of 2 (see Section~\ref{sec:anySizeLB}).
We use the optimal algorithm for zero-size items and the algorithm Pseudo to design an
(absolutely) $3.5$-competitive algorithm which is also
(asymptotically) $(1.5+d/(d-1))$-competitive in the parametric case,
where the items have sizes of at most $1/d$ for a real $d \geq 2$
(see Section~\ref{sec:anySizeUB}). (Note
that for $d<2$ we have $d/(d-1)>2$ and the bound for arbitrary items
is better.)

We show that algorithms BF, FF and WF are not
constant competitive, in contrast to their $3$-competitiveness for two colors.
Their competitiveness cannot be
bounded by any function of the number of colors even for only three
colors and very small items (see Section~\ref{sec:AFalgsOnColors}).

For Black and White Bin Packing,
we improve the upper bound on the
absolute competitive ratio of Any Fit algorithms in the general case
to $3$ which is tight for BF, FF and WF (see Section~\ref{sec:AF-BWBP}).  For WF in the parametric
case, we prove that it is absolutely $(1+d/(d-1))$-competitive for a real $d\geq 2$
which is tight for an integral $d$ (see Section~\ref{sec:WF-BWBP}).
Therefore, WF has the same competitive ratio as the Pseudo algorithm.

\medskip

\textbf{Related work.}
In the classical Bin Packing problem, we are given items with sizes in
$(0,1]$ and the goal is to assign them into the minimum number of
unit capacity bins. The problem was proposed by
Ullman~\cite{ullman71} and by Johnson~\cite{johnson73} and it is
known to be NP-hard. See the survey of Coffman et al.~\cite{coffman13} for the many
results on classical Bin Packing and its many variants.

For the online problem, there is no online algorithm which is better
than $248/161\approx 1.540$-competitive~\cite{balogh12}.  The
currently best algorithm is Harmonic++ by Seiden~\cite{seiden02},
approximately $1.589$-competitive.  Regarding AF algorithms, NF is
2-competitive and both FF and BF have the absolute competitive ratio
exactly 1.7~\cite{dosa13,dosa14}.  This is similar to Black and White
Bin Packing in which FF and BF have the absolute competitive ratio of 3
and the hard instances proving tightness of the bound are the same for
both algorithms.

In the context of Colored Bin Packing, we are interested in variants
that further restrict the allowed packings. Of particular interest is
Bounded Space Bin Packing where an algorithm can have only $K \geq 1$
open bins in which it is allowed to put incoming items.  When a bin is
closed an algorithm cannot pack any further item in the bin or open it
again.  Such algorithms are called \textit{$K$-bounded-space}. The
champion among these algorithms is $K$-Bounded Best Fit, i.e., Best
Fit with at most $K$ open bins, which is (asymptotically)
$1.7$-competitive for all $K \geq 2$~\cite{csirik01}.  Lee and
Lee~\cite{lee85} presented Harmonic($K$) which is $K$-bounded-space
with the asymptotic ratio of $1.691$ for $K$ large enough. Lee and Lee also
proved that there is no bounded space algorithm with a better asymptotic
ratio.

The Bounded Space Bin Packing is an especially interesting variant in
our context due to the fact that it matters whether we allow the
optimum to reorder the input instance or not. If we allow reordering
for Bounded Space Bin Packing, we get the same optimum as classical
Bin Packing. In fact, all the bounds on online algorithms in the
previous paragraph hold if the optimum with reordering is considered,
which is a stronger statement than comparing to the optimum without
reordering. This is a very different situation than for Colored Bin
Packing, where no online algorithms can be competitive against the
optimum with reordering, as we have noted above.

The bounded space offline optimum without reordering was studied by
Chro\-bak et al.~\cite{chrobak11}. It turns out that the computational
complexity is very different: There exists an offline
$(1.5+\varepsilon)$-approximation algorithm for 2-bounded-space Bin
Packing with polynomial running time for every constant
$\varepsilon>0$, but exponential in $\varepsilon$. No polynomial time
2-bounded-space algorithm can have its approximation ratio better than
$5/4$ (unless $P=NP$). In the online setting it is open whether there exists a
better algorithm than $1.7$-competitive $K$-Bounded Best Fit when
compared to the optimum without reordering; the current lower bound is
$3/2$.

Another interesting variant with restrictions on the contents of a bin
is Bin Packing with Cardinality Constraints, which restricts
the number of items in a bin to at most $k$ for a parameter $k \geq
2$. It was introduced by Krause et al.~\cite{krause75} who also showed
that Cardinality Constrained FF has the asymptotic ratio of at most
$2.7-2.4/k$. Interestingly, the lower bound for the asymptotic
competitive ratio for large $k$ is $1.540$~\cite{balogh12}, i.e., the
same as for standard Bin Packing. For $k\geq 3$, there is an
asymptotically $2$-competitive online algorithm~\cite{babel04} and the
absolute competitive ratio is at least 2 for any $k\geq
4$~\cite{dosa14bpcc}.  Better algorithms and various lower bounds are
known for small $k$~\cite{epstein06,babel04}.

\medskip

\textbf{Motivation.}
Suppose that a television or a radio station maintains several
channels and wants to assign a set of programs to them. The programs
have types like ``documentary'', ``thriller'', ``sport'' on TV,
or music genres on radio. To have a fancy schedule
of programs, the station does not want to broadcast two programs of the
same type one after the other. Colored Bin Packing can be used to
create such a schedule. Items here correspond to programs, colors to
genres and bins to channels.  Moreover, the programs can appear online
and have to be scheduled immediately, e.g., when listeners send
requests for music to a radio station via the Internet.

Another application of Colored Bin Packing comes from software which
renders user-generated content (for example from the
Internet) and assigns it to columns which are to be displayed.
The content is in boxes of different colors and we do not want
two boxes of the same color to be adjacent in a column,
otherwise they would not be distinguishable for the user.

Moreover, Colored Bin Packing with all items of size zero corresponds
to a situation in which we are not interested in loads of bins
(lengths of the schedule, sizes of columns, etc.), but we just want
some kind of diversity or colorfulness.

\section{Preliminaries and Offline Optimum}\label{sec:Prelim}

\textbf{Definitions and notation.}
There are three settings of Colored Bin Packing:
In the \textit{offline setting} we are given the items in advance
and we can pack them in an arbitrary order.
In the \textit{restricted offline setting} we also know sizes and colors of all items in advance,
but they are given as a sequence and they need to be packed in that order.
In the \textit{online setting} the items are coming one by one and we do not
know what comes next or even the total number of items. 
Moreover, an online algorithm has to pack each incoming item immediately
and it is not allowed to change its decisions later.

We focus mostly on the online setting.  To measure the effectiveness
of online algorithms for a particular instance $L$, we use the
restricted offline optimum denoted by $\vari{OPT}(L)$ or $\vari{OPT}$
when the instance $L$ is obvious from the context.  Let
$\vari{ALG}(L)$ denote the number of bins used by the algorithm
$\vari{ALG}$.  The algorithm is \textit{absolutely $r$-competitive} if
for any instance $\vari{ALG}(L)\leq r\cdot{\vari{OPT}(L)}$ and
\textit{asymptotically $r$-competitive} if for any instance
$\vari{ALG}(L)\leq r\cdot{\vari{OPT}(L)}+o(\vari{OPT}(L))$;
typically the additive term is just a constant.  We say that an
algorithm has the (absolute or asymptotic) competitive ratio $r$
if it is (absolutely or asymptotically) $r$-competitive and it is
not $r'$-competitive for $r'<r$.

For Colored Bin Packing, let $C$ be the set of all colors. For $c \in
C$, the items of color $c$ are called $c$-items and bins with the top (last)
item of color $c$ are called $c$-bins. By a non-$c$-item we mean an item of
color $c'\neq c$ and similarly a non-$c$-bin is a bin of 
color $c'\neq c$.  The \textit{level of a bin} means the cumulative size
of all items in the bin.

We denote a sequence of $nk$ items consisting of $n$ groups of $k$ items of colors $c_1, c_2, \dots c_k$ and sizes
$s_1, s_2, \dots s_k$ by $n \times \left( \nolinefrac{c_1}{s_1}, \nolinefrac{c_2}{s_2}, \dots \nolinefrac{c_k}{s_k} \right)$.

\medskip

\textbf{Lower Bounds on the Restricted Offline Optimum.}
We use two lower bounds on the number of bins in any packing. The first bound
$\vari{LB_1}$ is the sum of sizes of all items.

The second bound $\vari{LB_2}$ is the maximal color discrepancy inside
the input sequence.  In Black and White Bin Packing, the color
discrepancy introduced by Balogh et al.~\cite{balogh14} is simply the
difference of the number of black and white items in a segment of
the input sequence, maximized over all segments.  It is easy to see that
it is a lower bound on the number of bins.

In the generalization of the color discrepancy for more than two colors we
count the difference between $c$-items and non-$c$-items for all
colors $c$ and segments. It is easy to see that this is a lower bound
as well. Formally, let $s_{c,i}=1$ if the $i$-th item from the input
sequence has color $c$, and $s_{c,i}=-1$ otherwise. We define
\[
\vari{LB_2} = \max_{c \in C} \max_{i, j} \sum_{\ell=i}^{j}
s_{c,\ell}\,.
\]
For Black and White Bin Packing, equivalently 
$\vari{LB_2} = \max_{i, j} \left| \sum_{\ell=i}^{j} s_\ell \right|$,
where $s_i=1$ if the $i$-th item is white, and $s_i=-1$ otherwise; the
absolute value replaces the maximization over colors.

We prove that $\vari{LB_2}$ is a lower bound on the optimum similarly to the proof of Lemma~5 in~\cite{balogh14}.
First we observe that the number of bins in the optimum cannot increase by removing a prefix or a suffix
from the sequence of items.

\begin{observation}
Let $L = L_1 L_2 L_3$ be a sequence of items partitioned into three subsequences (some of them can be empty).
Then $\vari{OPT}(L) \geq \vari{OPT}(L_2).$ 
\end{observation}

\begin{proof}
It is enough to show that the removal of the first or the last item does not increase the optimum. 
By iteratively removing items from the beginning and the end of the sequence we obtain the subsequence $L_2$
and consequently $\vari{OPT}(L) \geq \vari{OPT}(L_2)$.

The first item of the sequence is clearly the first item in a bin. By removing the first item from the bin we do not violate any condition.
Hence any packing of $L$ is a valid packing of $L$ without the first item.
A similar claim holds for the last item.
\end{proof}

\begin{lemma}
$\vari{OPT}(L) \geq \vari{LB_2}.$ 
\end{lemma}

\begin{proof}
We prove that for all colors $c$ that the optimum is at least $\vari{LB_{2,c}} := \max_{i, j} \sum_{\ell=i}^{j} s_{c,\ell}$.
Fix a color $c$ and let $i, j$ be $\argmax_{i, j} \sum_{\ell=i}^{j} s_{c,\ell}$.
Let $\delta = \vari{LB_{2,c}}$. We may assume that $\delta > 0$, otherwise $\delta$ is trivially at most the optimum.
By the previous observation we may assume $i = 1$ and $j = n$.

Consider any packing of the sequence and let $k$ be the number of bins used.
Any bin contains at most one more $c$-item than non-$c$-items,
since colors are alternating between $c$ and other colors in the worst case.
Since we have $\delta$ more $c$-items than non-$c$-items, we get $k\geq \delta$.
Therefore $\vari{OPT} \geq \vari{LB_{2,c}}$
\end{proof}

In Black and White Bin Packing, when all the items are of size zero,
all Any Fit algorithms create a packing into the optimal number of
bins~\cite{balogh14}.  For more than two colors this is not true and
in fact no deterministic online algorithm can have a competitive ratio
below $1.5$.  However, in the restricted offline setting a packing
into $\vari{LB_2}$ bins is still always possible, even though this fact
is not obvious. This shows that the color discrepancy fully characterizes the
combinatorial aspect of the color restriction in Colored Bin Packing.

\begin{theorem}
\label{thm:LB2feasible}
Let all items have size equal to zero.
Then a packing into $\vari{LB_2}$ bins is possible in the restricted offline setting,
i.e., items can be packed into $\vari{LB_2}$ bins without reordering.
\end{theorem}

\begin{proof}
Consider a counterexample with a minimal number of items in the sequence.
Let $d = \vari{LB_2}$ be the maximal discrepancy in the counterexample
and $n \geq d$ be the number of items.
The minimality implies that the theorem holds for all sequences of length $n' < n$.
Moreover, $d > 1$, since for $d = 1$ we can pack the sequence trivially
into a single bin.

We define an \textit{important interval} as a maximal interval of discrepancy $d$,
more formally a subsequence from the $i$-th item to the $j$-th such that
for some color $c$ the discrepancy on the interval is $d$, i.e., $\sum_{\ell=i}^{j} s_{c,\ell} = d$,
and we cannot extend the interval in either direction without decreasing its discrepancy.
For an important interval, its \textit{dominant color} $c$ is the most frequent color inside
it. At first we show that important intervals are just $d$ items of the same color.

\begin{observation}
Each important interval $I$ contains only $d$ items of its dominant color $c$ in the minimal counterexample.
\end{observation}
\begin{proof}
Suppose there is a non-$c$-item in $I$ and let $a$ be the last such item in $I$.
Then $a$ must be followed by a $c$-item $b$ in $I$, otherwise $I$ without $a$ would have higher discrepancy.
We delete $a$ and $b$ from the sequence and pack the rest into $d$ bins by minimality.

Consider the situation after packing the item prior to $a$.
There must be a $c$-bin $B$, otherwise the subsequence of $I$ from the beginning up to $a$ (including $a$)
has strictly more non-$c$-items than $c$-items
(each $c$-item from $I$ is under a non-$c$-item
and $a$ is the extra non-$c$-item).
Hence the rest of $I$ after $a$ must have discrepancy of more than $d$.
By putting $a$ and $b$ into $B$ we pack the whole sequence into $d$ bins,
thus it is not a counterexample.
\end{proof}

We show that important intervals are disjoint in the minimal counterexample.
Suppose that two important intervals $I_1$ and $I_2$ with dominant colors $c_1$ and $c_2$
intersect on an interval $J$. If $c_1 \neq c_2$ we use the previous observation,
since $I_1$ or $I_2$ has to contain an item from the other interval.
Otherwise $c_1 = c_2$, but then their union has discrepancy of more than $d$
which cannot happen.

Clearly, there must be an important interval in any non-empty sequence.
Let $I_1, I_2, \dots I_k$ be important intervals in the counterexample sequence and
let $J_1, J_2, \dots J_{k - 1}$ be the intervals between the important intervals ($J_i$ between $I_i$ and $I_{i + 1}$),
$J_0$ be the interval before $I_1$ and $J_k$ be the interval after $I_k$. 
These intervals are disjoint and form a complete partition of the sequence, i.e.,
$J_0, I_1, J_1, I_2, J_2, \dots$ $J_{k - 1}, I_k, J_k$ is the whole sequence of items.
Note that some of $J_\ell$'s can be empty.

If $k > 2$, we can create a packing $P_1$ of the sequence containing
only intervals $J_0, I_1, J_1, I_2$ into $d$ bins
by minimality of the counterexample.
Also there exists a packing $P_2$ of intervals $I_2, J_2, I_3, \dots I_k, J_k$ into $d$ bins.
Any bin from $P_1$ must end with an item from the important interval $I_2$ and any bin from $P_2$ must start
with an item from $I_2$. Therefore we can merge both packings by items
from $I_2$ and obtain a valid packing of the whole sequence into $d$
bins. Hence $k \leq 2$.

In the case $k = 1$, there are four subcases depending on whether $J_0$ and $J_1$ are empty or not:
\begin{compactitem}
\item $J_0$ and $J_1$ are non-empty: We create packings of $J_0, I_1$ and $I_1, J_1$ into $d$ bins
  and merge them as before.
\item $J_0$ is empty and $J_1$ non-empty: We delete the first item from $I_1$, pack the rest into $d - 1$ bins
  (the maximal discrepancy decreases after deleting) and put the deleted item into a separate bin.
\item $J_0$ is non-empty and $J_1$ empty: Similarly, we delete the last item from $I_1$ and pack the rest into $d - 1$ bins.
\item both are empty: $I_1$ can be trivially packed into $d$ bins.
\end{compactitem}

For $k = 2$, we first show that $J_0$ and $J_2$ are empty and $J_1$ is non-empty in the counterexample.
If $J_0$ is non-empty, we merge packings of $J_0, I_1$ and $I_1, J_1, I_2, J_2$,
and if $J_2$ is non-empty, we put together packings of $J_0, I_1, J_1, I_2$ and $I_2, J_2$.
When $J_1$ is empty, the sequence consists only of intervals $I_1$ and $I_2$ which must have different
dominant colors. Thus they can be easily packed one on the other into $d$ bins.

The last case to be settled has only $I_1, J_1$ and $I_2$ non-empty. If the dominant colors $c_1$ for $I_1$
and $c_2$ for $I_2$ are different, we delete the first item from $I_1$
and the last item from $I_2$, so the discrepancy decreases.
We pack the rest into $d-1$ bins and put the deleted items into a
separate bin, so the whole sequence is in $d$ bins again.

Otherwise $c_1$ is equal to $c_2$ and let $c$ be $c_1$. Since the important intervals are maximal, there must
be at least $d + 1$ more non-$c$-items than $c$-items in $J_1$.
Also any prefix of $J_1$ contains strictly more non-$c$-items than $c$-items,
thus at least the first two items in $J_1$ have colors different from $c$.

We delete the first $c$-item $p$ from $I_1$, the first non-$c$-item $q$ from $J_1$
and the last $c$-item $r$ from $I_2$.
Suppose for a contradiction that there is an interval $I$ of discrepancy
$d$ in the rest of the sequence. As $I$ has lower discrepancy in the original sequence (we deleted an item from each important interval of the original sequence),
it must contain $q$ and thus intersect $I_1$ and $J_1$, hence its dominant color is $c$.
If $I$ intersects also $I_2$, we add the items $p, q$ and $r$ into $I$
(and possibly some other items from $I_1$ or $I_2$)
to obtain an interval of discrepancy at least $d+1$ in the original sequence which is a contradiction.
Otherwise $I$ intersects only $I_1$ and $J_1$, but any prefix of the rest
of $J_1$ still contains at least as many non-$c$-items as $c$-items,
so $I \setminus J_1$ has discrepancy at least $d$. But $I \setminus J_1$ is contained in 
the rest of $I_1$ that has only $d-1$ items and we get a contradiction.
Therefore the maximal discrepancy decreases after deleting the three items,
so we can pack the rest into $d - 1$ bins and the items $p, q$ and $r$ are put into a separate bin.
Note that important intervals of discrepancy $d - 1$ may change after deleting the three items.

In all cases we can pack the sequence into $d$ bins, therefore no such counterexample exists.
\end{proof}

\section{Algorithms for Zero-size Items}\label{sec:zeroSize}

\subsection{Lower Bound on Competitiveness of Any Online Algorithm}\label{sec:zeroSizeLB}

\begin{theorem}
\label{thm:LBforColors}
For zero-size items of at least three colors,
there is no deterministic online algorithm with an asymptotic competitive ratio less than $1.5$.
Precisely, for each $n>1$ we can force any deterministic online algorithm to use at least $\lceil 1.5n \rceil$ bins using three colors,
while the optimal number of bins is $n$.
\end{theorem}

\begin{proof}
We show that if an algorithm uses less than $\lceil 1.5n \rceil$ bins,
the adversary can send some items and force the algorithm to increase the
number of black bins or to use at least $\lceil 1.5n \rceil$ bins,
while the maximal discrepancy stays $n$.
Applying Theorem~\ref{thm:LB2feasible} we know that $\vari{OPT} = n$,
but the algorithm is forced to open $\lceil 1.5n \rceil$ bins using
finitely many items as the number of black bins is increasing.
Moreover, the adversary uses only three colors throughout the whole proof, denoted by black, white and red
and abbreviated by b, w and r in formulas.

We introduce the current discrepancy of a color $c$ which basically tells us how many $c$-items
have come recently and thus how many $c$-items may arrive without increasing the overall discrepancy.
Formally, we define the current discrepancy after packing the $k$-th item as
$\vari{CD}_{c,k} = \max_{i \leq k+1} \sum_{\ell=i}^k s_{c,\ell}$,
i.e., the discrepancy on an interval which ends with the last packed item (the $k$-th).
Note that $\vari{CD}_{c,k}$ is at least zero as we can set $i = k+1$.
We omit the $k$ index in $\vari{CD}_{c,k}$ when it is obvious from the context.

Initially the adversary sends $n$ black items, then he continues by phases and
ends the process whenever the algorithm uses $\lceil 1.5n \rceil$ bins at the end of a phase.
When a phase starts, there are less than $\lceil 1.5n \rceil$
black bins and possibly some other white or red bins.
We also guarantee $\vari{CD}_{\mathrm{w}} = 0$,
$\vari{CD}_{\mathrm{r}} = 0$, and $\vari{CD}_{\mathrm{b}} \leq n$. 
Let $N_{\mathrm{b}}$ be the number of black bins when a phase starts.
In each phase the adversary forces the algorithm to use $\lceil 1.5n \rceil$ bins
or to have more than $N_{\mathrm{b}}$ black bins,
while $\vari{CD}_{\mathrm{w}} = 0$, $\vari{CD}_{\mathrm{r}} = 0$, and
$\vari{CD}_{\mathrm{b}} \leq n$ at the end of each phase in which $N_{\mathrm{b}}$ increases.

We now present how a phase works. Let new items be items from the current phase and
old items be items from previous phases.
The adversary begins the phase by sending $n$ new items of colors alternating between
white and red, starting by white, so he sends $\lceil n/2\rceil$
white items and $\lfloor n/2\rfloor$ red items.
After these new items, the current discrepancy is one either for red if $n$ is even,
or for white if $n$ is odd, and it is zero for the other colors.

If some new item is not put on an old black item, 
the adversary sends $n$ black items. Since the new items are packed into less
than $n$ black bins (more precisely, black at the beginning of the phase),
the number of black bins increases. Moreover, $\vari{CD}_{\mathrm{w}} = 0$,
$\vari{CD}_{\mathrm{r}} = 0$, and $\vari{CD}_{\mathrm{b}} = n$,
hence the adversary finishes the phase and continues with the next phase if there are
less than $\lceil 1.5n \rceil$ black bins.

Otherwise all new red and white items are put on old black items. If
$n$ is even, $\vari{CD}_{\mathrm{w}} = 0$ and the adversary sends additional $n$
white items. After that there are at least $1.5n$ white bins, so the adversary
reaches his goal.

If $n$ is odd, $\vari{CD}_{\mathrm{w}} = 1$ and the adversary
can send only $n-1$ white items forcing $\lceil 1.5n\rceil - 1$ white bins.
This suffices to prove the result in the asymptotic sense,
but for the precise lower bound of $\lceil 1.5n\rceil$
for an odd $n$ we need a somewhat more complicated construction.

Therefore if all new red and white items are put on old black items and
$n$ is odd, the adversary sends a black item $e$.
If $e$ does not go on a new white item, he sends $n$ white items
forcing $\lceil n/2\rceil + n$ white bins and it is done.
Otherwise the black item $e$ is put on a new white item.
White and red have $\lfloor n/2\rfloor$ new items on the top of bins,
$\vari{CD}_{\mathrm{w}} = 0$, and $\vari{CD}_{\mathrm{r}} = 0$. 
The adversary sends another black item $f$.
Since red and white are equivalent colors (considering only new items),
w.l.o.g.\ $f$ goes into a red bin or into newly opened bin.

Next he sends a white item $g$ and a red item $h$.
After packing $g$ there are $\lceil n/2\rceil$ bins with a new
white item on the top and at least one bin with a new black item on the top.
Moreover, after packing the red item $h$ we have
$\vari{CD}_{\mathrm{b}} = 0$ and $\vari{CD}_{\mathrm{w}} = 0$.
So if $h$ is not put on a new white item (i.e., it is put into a black
bin, a new bin or on an old white item), the adversary sends $n$ white items
and the algorithm must use $\lceil 1.5n\rceil$ bins.
Otherwise $h$ is packed on a new white item and the adversary sends $n$ black items.
The number of black bins increases, because the adversary sent $n+2$ new black
items and at most $n+1$ new non-black items were put into a black bin
(at most $n$ items at the beginning of the phase plus the item $g$).
Since $\vari{CD}_{\mathrm{w}} = 0$, $\vari{CD}_{\mathrm{r}} = 0$, and
$\vari{CD}_{\mathrm{b}} = n$,
the adversary continues with the next phase.
\end{proof}

The lower bound has additional properties that we use later in our lower
bound for items of arbitrary size. Most importantly, we have at least
$\lceil 1.5\cdot \vari{OPT} \rceil$ of $c$-bins at the end (and
possibly some additional bins of other colors).

\begin{lemma} \label{l:LBforColorsInstance}
After packing the instance from Theorem~\ref{thm:LBforColors} by an online algorithm
there is a color $c$
for which we have $\lceil 1.5\cdot \vari{OPT} \rceil$ of $c$-bins and $\vari{CD}_{c} = \vari{OPT}$,
while $\vari{CD}_{c'} = 0$ for all other colors $c' \neq c$.
Moreover, in each restricted offline optimal packing of the instance all the bins have a $c$-item on the top.
\end{lemma}

\begin{proof}
Let $n = \vari{OPT}$ as in the previous proof. 
The adversary stops sending items when he finishes the last phase. In the last phase either
the number of black bins increases to $\lceil 1.5n\rceil$, or the adversary forces $\lceil 1.5n\rceil$
white or red bins by sending $n$ white or red items. In the former case the requirements
of the lemma are satisfied, because the proof guarantees $\vari{CD}_{\mathrm{w}} = 0$ and $\vari{CD}_{\mathrm{r}} = 0$
at the end of each phase in which the number of black bins increases. Moreover $\vari{CD}_{\mathrm{b}} = n$,
since $n$ black items are sent just before the end of such phase.
In the latter case, the last $n$ white items cause $\vari{CD}_{\mathrm{b}} = 0$, $\vari{CD}_{\mathrm{r}} = 0$, and
$\vari{CD}_{\mathrm{w}} = n$; the case of $n$ red items is symmetric.

Since an optimal packing uses $n$ bins and the last $n$ items are of the same color (in each case of the construction), 
they must go into different bins. Hence each bin of a restricted offline optimal packing has a $c$-item on the top.
\end{proof}

\subsection{Optimal Algorithm for Zero-size Items}\label{sec:zeroSizeUB}

The overall problem of FF, BF and WF is that they pack items regardless of the colors of bins.
We address the problem by balancing the colors of top items in bins
-- we mostly put an incoming $c$-item into a bin of
the most frequent other color.
When we have more choices of bins where to put an item we use First Fit.
We call this algorithm \textit{Balancing Any Fit} (BAF).

We define BAF for items of size zero
and show that it opens at most $\lceil 1.5 \vari{LB_2}\rceil$ bins
which is optimal in the worst case by Theorem~\ref{thm:LBforColors}.
Then we combine BAF with the algorithm Pseudo by Balogh et al.~\cite{balogh14} for items of arbitrary size
and prove that the resulting algorithm is absolutely $3.5$-competitive.

After packing the $k$-th item from the sequence, let $D_k$ be the maximal discrepancy so far, i.e., the discrepancy on
an interval before the $(k+1)$-st item, and let $N_{c,k}$ be the number of $c$-bins after packing the $k$-th item.
As in the proof of Theorem~\ref{thm:LBforColors}, we define the current discrepancy as
$\vari{CD}_{c,k} = \max_{i \leq k+1} \sum_{\ell=i}^k s_{c,\ell}$,
i.e., the discrepancy on an interval which ends with the last packed item (the $k$-th).
Note that $\vari{CD}_{c,k} \leq D_k$ and that $\vari{CD}_{c,k}$ is at least zero as we can set $i = k+1$.
The current discrepancy basically tells us how many $c$-items
have come recently and thus how many $c$-items may arrive without increasing the overall discrepancy.

Let $\alpha_{c,k} = N_{c,k} - \lceil D_k/2\rceil $ be the difference
between the number of $c$-bins and the half of the maximal discrepancy so far.
Observe that $\lceil D_k/2\rceil $ is the number of bins which BAF may use in addition to $\vari{OPT}$ bins.
We omit the index $k$ in $D_k$, $N_{c,k}$, $\vari{CD}_{c,k}$ and $\alpha_{c,k}$
when it is obvious from the context.

While processing the items,
if $D$ is the maximal discrepancy so far,
the adversary can send $D - \vari{CD}_c$ of $c$-items
without increasing the maximal discrepancy and forcing the algorithm to use $N_c + D - \vari{CD}_c$ bins.
Hence, to end with at most $\lceil 1.5D\rceil$ bins we try to keep
$N_c - \vari{CD}_c \leq \lceil D/2\rceil$
for all colors $c$. For simplicity, we use an equivalent inequality of 
$\alpha_c = N_c - \lceil D/2\rceil \leq \vari{CD}_c$.
If we can keep the inequality valid and it occurs that there is a color $c$ with
$N_c > \lceil 1.5D\rceil$, we get
$\vari{CD}_c \geq N_c - \lceil D/2\rceil > \lceil 1.5D\rceil - \lceil D/2\rceil = D$
which contradicts $\vari{CD}_c \leq D$.
Let the \textbf{main invariant} for a color $c$ be
\begin{equation} \label{eq:mainInvariant}
\alpha_c = N_c - \left\lceil\frac{D}{2}\right\rceil \leq \vari{CD}_c.
\end{equation}

As $\vari{CD}_c \geq 0$, keeping the invariant is easy for all colors
with at most $\lceil D/2\rceil$ bins. Also when there is only one
color $c$ with $N_c > \lceil D/2\rceil$, we just put all non-$c$-items
into $c$-bins. Therefore, if a non-$c$-item comes, the number of $c$-bins $N_c$ decreases and the
current discrepancy $\vari{CD}_c$ decreases by at most one. ($\vari{CD}_c$ stays the same when it is zero.)
Since both increase with an incoming $c$-item, we are keeping our
main invariant~(\ref{eq:mainInvariant}) for the color $c$. 

Moreover, there are at most two colors with strictly more than $\lceil
D/2\rceil$ bins, given that we have at most $\lceil 1.5D\rceil$ open
bins. Thus we only have to deal with two colors having $N_c > \lceil
D/2\rceil$. We state the algorithm Balancing Any Fit for items of
size zero.

\begin{center}
\fbox{\parbox{0.9\textwidth}{
{\bf Balancing Any Fit (BAF):}
\begin{compactenum}
\item For an incoming $c$-item, if there are no bins or $c$-bins only, open a new bin and put the item into it.
\item Otherwise, if there is at most one color with the number of bins strictly more than $\lceil D/2\rceil$, put an incoming $c$-item
into a bin of color $c' = \argmax_{c'' \neq c} N_{c''}$.
If more colors have the same maximal number of bins, choose color $c'$ arbitrarily among them, e.g., by First Fit.
Among $c'$-bins, choose again arbitrarily.
\item Suppose that there are two colors $\mathrm{b}$ and $\mathrm{w}$
  such that $N_{\mathrm{b}} > \lceil D/2\rceil$ and $N_{\mathrm{w}} >
  \lceil D/2\rceil$. 
If $c=\mathrm{w}$, put the item into a bin of color $\mathrm{b}$.
If $c=\mathrm{b}$, put the item into a bin of color $\mathrm{w}$.
Otherwise $c\not\in\{\mathrm{b},\mathrm{w}\}$; if $N_{\mathrm{b}} -
\lceil D/2\rceil < \vari{CD}_{\mathrm{b}}$, put the item into a bin
of color $\mathrm{w}$, otherwise into a bin of color $\mathrm{b}$.
\end{compactenum}
}}
\end{center}

As we discussed, keeping the main
invariant~(\ref{eq:mainInvariant}) is easy in the first and the second
case of the algorithm. Therefore we can conclude the following claim.

\begin{claim} \label{clm:OptAlgMainInvar}
Suppose that the main invariant holds for all colors before packing the $t$-th item
and that there is at most one color $c$ with $N_{c,t-1} > \lceil D_{t-1}/2\rceil$ before the $t$-th item,
i.e., the $t$-th item is packed using the first or the second
case of the algorithm. Then the main invariant holds for all colors also after packing the $t$-th item.
\end{claim}

Most of the proof of $1.5$-competitiveness of BAF thus deals with two colors having
more than $\lceil D/2\rceil$ bins. W.l.o.g.\ let
these two colors be black and white in the following and let us abbreviate them by
${\mathrm{b}}$ and ${\mathrm{w}}$.

In the third case of the algorithm we have to choose either black or
white bin for items of other colors than black and white, but the
current discrepancy decreases for both black and white, while the
number of bins stays the same for the color which we do not choose.
So if $\alpha_{\mathrm{b}} = \vari{CD}_{\mathrm{b}}$ and
$\alpha_{\mathrm{w}} = \vari{CD}_{\mathrm{w}}$, the adversary can
force the algorithm to open more than $\lceil 1.5D\rceil$ bins.

Therefore we need to prove that in the third case, i.e., when
$N_{\mathrm{b}} > \lceil D/2\rceil$ and $N_{\mathrm{w}} > \lceil
D/2\rceil$, at least one of inequalities $\alpha_{\mathrm{b}} \leq
\vari{CD}_{\mathrm{b}}$ and $\alpha_{\mathrm{w}} \leq
\vari{CD}_{\mathrm{w}}$ is strict. This motivates the following
\textbf{secondary invariant:} 
\begin{equation} \label{eq:secondaryInvariant}
2\alpha_{\mathrm{b}} + 2\alpha_{\mathrm{w}} \leq \vari{CD}_{\mathrm{b}} + \vari{CD}_{\mathrm{w}} + 1\,.
\end{equation}
If the secondary invariant holds, it is not hard to see that in the third
case of the algorithm the choice of the bin maintains the main
invariant. The tricky part of the proof is to prove the {\em base
case} of the inductive proof of the secondary invariant.
A natural proof would show the base case
whenever ${\mathrm{b}}$ and ${\mathrm{w}}$ become the two colors with
$N_{\mathrm{b}},N_{\mathrm{w}} > \lceil D/2\rceil$. However, we are
not able to do that. Instead, 
we prove that the secondary invariant
holds already at the moment when ${\mathrm{b}}$ and ${\mathrm{w}}$
become the two strictly most frequent colors on top of the bins, i.e.,
$N_{\mathrm{b}} > N_c$ and $N_{\mathrm{w}} > N_c$ for all other colors $c$, which may
happen much earlier, when the number of their bins is significantly below $D/2$. 
After that, maintaining both invariants is relatively easy.

\begin{theorem} \label{thm:optAlg}
Balancing Any Fit is $1.5$-com\-pe\-ti\-ti\-ve for items of
size zero and an arbitrary number of colors. Precisely, it uses at most
$\lceil 1.5 \cdot \vari{OPT}\rceil$ bins.
\end{theorem}

\begin{proof} 
First we show that keeping the main invariant~(\ref{eq:mainInvariant})
for each color $c$, i.e., $\alpha_c \leq \vari{CD}_c$, is sufficient
for the algorithm to create at most $\lceil 1.5D\rceil$ bins. This
implies both that the algorithm is well defined since there are at
most two colors with $N_c > \lceil D/2\rceil$, and that the algorithm
is $1.5$-competitive, since the maximal discrepancy equals the optimum.

\begin{claim}
\label{cl:BAFineqForColors}
After packing the $t$-th item,
if we suppose that $N_{c,i} - \lceil D_i/2\rceil \leq \vari{CD}_{c,i}$ for all colors $c$ and for all $i<t$,
the algorithm uses at most $\lceil 1.5D_t\rceil$ bins.
\end{claim}

\begin{proof}
We prove the claim by contradiction:
Suppose that BAF opens a bin with the $k$-th item
in the sequence (for $k \leq t$) and we exceed the $\lceil 1.5D_k\rceil$ limit,
but before the $k$-th item there were at most $\lceil 1.5D_{k-1}\rceil$ bins.
Thus $D_k = D_{k-1}$, since if $D_k = D_{k-1} + 1$, then the
bound also increases with the $k$-th item.

Let $c$ be the color of the $k$-th item. Let the $\ell$-th item be the last non-$c$-item before the $k$-th, so only
$c$-items come after the $\ell$-th item.
None of $c$-items from the $(\ell + 1)$-st to the $k$-th increase the maximal discrepancy $D$,
otherwise if one such item increases $D$, then all following such items also do.
Thus $D_{\ell} = D_k$.

The algorithm must have received $\lceil 1.5D_{\ell} \rceil + 1 - N_{c,\ell}$ of
$c$-items after the $\ell$-th item to open $\lceil 1.5D_{\ell} \rceil + 1$ bins, but then
$$\vari{CD}_{c,k} = \vari{CD}_{c,\ell} + \lceil 1.5D_{\ell} \rceil + 1 - N_{c,\ell}
\geq N_{c,\ell} - \left\lceil\frac{D_{\ell}}{2}\right\rceil + \lceil 1.5D_{\ell} \rceil + 1 - N_{c,\ell}
= D_{\ell} + 1$$
where we used the main invariant for the inequality which holds, because $\ell < k \leq t$.
We get a contradiction, since $\vari{CD}_{c,k} \leq D_{k} = D_{\ell}$.
\end{proof}

We have to deal with the case in which $N_{\mathrm{b}} > \lceil D/2\rceil$ and
$N_{\mathrm{w}} > \lceil D/2\rceil$.  We show that we can maintain the
secondary invariant (\ref{eq:secondaryInvariant}), while black and
white are the two strictly most frequent colors of bins
(even if $N_{\mathrm{b}} \leq \lceil D/2\rceil$ or
$N_{\mathrm{w}} \leq \lceil D/2\rceil$). Then we prove
that the secondary invariant starts to hold when black and white
become the two strictly most frequent colors, i.e.,
$N_c < N_{\mathrm{b}}$ and $N_c < N_{\mathrm{w}}$ for all other colors $c$;
this step must precede the time when the number of bins for the second color gets over the
$\lceil D/2\rceil$ limit.
Therefore we prove by induction that the secondary invariant holds
in certain intervals of the input sequence.

\begin{claim} \label{clm:OptAlgInductionSecondary}
Suppose that black and white are the two strictly most frequent colors of bins
before packing the $t$-th item and that the main invariant~(\ref{eq:mainInvariant}) holds for all
colors and the secondary invariant~(\ref{eq:secondaryInvariant}) also holds before packing the $t$-th item, i.e., 
$N_{c,t-1} - \lceil D_{t-1}/2\rceil \leq \vari{CD}_{c,t-1}$ for all colors $c$ and
$2\alpha_{\mathrm{b},t-1} + 2\alpha_{\mathrm{w},t-1} \leq \vari{CD}_{\mathrm{b},t-1} + \vari{CD}_{\mathrm{w,t-1}} + 1$.
Then the main invariant for all colors and the secondary invariant for black and white hold also after packing the $t$-th item.
\end{claim}

\begin{proof}
First we suppose that the maximal discrepancy $D$ is not changed by the $t$-th item.
We start by showing that the main invariant holds after packing the $t$-th item.
If the $t$-th item is packed using the second case of BAF,
the main invariant holds by Claim~\ref{clm:OptAlgMainInvar}.
(Note that the $t$-th item cannot be packed using the first case of the algorithm,
since $N_{\mathrm{b},t-1} > 0$ and $N_{\mathrm{w},t-1} > 0$.)

Otherwise, if the $t$-th item is packed using the third case,
it holds that $\alpha_{\mathrm{b},t-1} > 0$ and $\alpha_{\mathrm{w},t-1} > 0$.
The main invariant holds for a color $c$ other than black and white,
because $N_{c,t-1} < \lceil D_{t-1}/2\rceil$ which implies $N_{c,t}\leq \lceil D_{t}/2\rceil$.

To prove the main invariant for black and white,
we show by contradiction that the secondary invariant~(\ref{eq:secondaryInvariant})
guarantees that
$\alpha_{\mathrm{b},t-1} < \vari{CD}_{\mathrm{b},t-1}$ or $\alpha_{\mathrm{w},t-1} < \vari{CD}_{\mathrm{w},t-1}$.
Otherwise, if $\alpha_{\mathrm{b},t-1} \geq \vari{CD}_{\mathrm{b},t-1}$ and
$\alpha_{\mathrm{w},t-1} \geq \vari{CD}_{\mathrm{w},t-1}$, the secondary invariant becomes
$2\alpha_{\mathrm{b},t-1} + 2\alpha_{\mathrm{w},t-1} \leq \vari{CD}_{\mathrm{b},t-1} + \vari{CD}_{\mathrm{w},t-1} + 1
\leq \alpha_{\mathrm{b},t-1} + \alpha_{\mathrm{w},t-1} + 1$ which is a contradiction.
Note that we used that $\alpha_{\mathrm{w},t-1}$ and $\alpha_{\mathrm{b},t-1}$
are integral and positive.

We now distinguish three cases according to the color of the $t$-th item:
\begin{compactitem}
\item The $t$-th item is black: Then it is packed into a white bin.
The main invariant for black holds after packing the item, because
both $N_{\mathrm{b}}$ and $\vari{CD}_{\mathrm{b}}$ increase, 
and the main invariant for white holds, since
$N_{\mathrm{w}}$ decreases and $\vari{CD}_{\mathrm{w}}$ decreases by at most one.
($\vari{CD}_{\mathrm{w}}$ stays the same when it is zero.)
\item When the $t$-th item is white, the situation is symmetric to the previous case.
\item We pack the $t$-th item of another color into a white bin if
$N_{\mathrm{b},t-1} - \lceil D_{t-1}/2\rceil < \vari{CD}_{\mathrm{b},t-1}$,
otherwise into a black bin. 
If it is packed into a white bin, $N_{\mathrm{w}}$ decreases and $\vari{CD}_{\mathrm{w}}$ decreases by at most one,
thus the main invariant holds for white.
The main invariant holds for black too, since $N_{\mathrm{b}}$ stays the same
and $\vari{CD}_{\mathrm{b}}$ decreases by at most one, but the main invariant held strictly for black
before packing the $t$-th item.

If the $t$-th item is packed into a black bin,
we have $N_{\mathrm{w},t-1} - \lceil D_{t-1}/2\rceil < \vari{CD}_{\mathrm{w},t-1}$
and the situation is symmetric as if the $t$-th item is packed into a white bin.
\end{compactitem}

It remains to show that the $t$-th item does not violate the secondary invariant.
There are again three cases according to the color of the $t$-th item:
\begin{compactitem}
\item The $t$-th item is black: Then it is packed into a white bin in both second and third cases of the algorithm.
Thus $\alpha_{\mathrm{b}}$ increases and $\alpha_{\mathrm{w}}$ decreases,
so the left-hand side of the inequality stays the same.
Also the right-hand side does not change or even increases as $\vari{CD}_{\mathrm{b}}$ increases
and $\vari{CD}_{\mathrm{w}}$ decreases by at most one. ($\vari{CD}_{\mathrm{w}}$ stays the same when it is zero.)
\item When the $t$-th item is white, the situation is symmetric to the previous case.
\item The $t$-th item has another color than black and white:
Then it is packed into a white or black bin in both second and third cases of the algorithm.
Thus one of $\alpha_{\mathrm{w}}$ and $\alpha_{\mathrm{b}}$ decreases and the other one stays the same,
while both $\vari{CD}_{\mathrm{b}}$ and $\vari{CD}_{\mathrm{w}}$ decrease by at most one.
The secondary invariant holds as the left-hand side decreases by two
and the right-hand side decreases by at most two.
\end{compactitem}

Otherwise $D$ increases with an incoming item, thus $\alpha_{c'}$ for each color $c'$ decreases
if $D$ becomes odd. We follow the same proof as if $D$ stays the same,
and the eventual additional decrease of $\alpha_{c'}$
can only decrease the left-hand sides of the main and secondary
invariants.
\end{proof}

Note that in the previous proof, $\alpha_{\mathrm{b}}$ or $\alpha_{\mathrm{w}}$
can be negative in the secondary invariant.
We show the base case of the secondary invariant, i.e., that it
starts to hold when two colors become the two strictly most frequent
colors of bins.

\begin{claim}  \label{clm:OptAlgBaseCaseSecondary}
Suppose that after packing the $k$-th item it starts to hold that
$N_c < N_{\mathrm{b}}$ and $N_c < N_{\mathrm{w}}$ for all other colors $c$, i.e.,
black and white become the two strictly most frequent colors.
Suppose also that the main invariant holds all the time before packing the $k$-th item.
Then $2\alpha_{\mathrm{b},k} + 2\alpha_{\mathrm{w},k} \leq \vari{CD}_{\mathrm{b},k} + \vari{CD}_{\mathrm{w},k} + 1$,
i.e., the secondary invariant holds after packing the $k$-th item.
\end{claim}

\begin{proof}
Assume w.l.o.g.\ that $N_{\mathrm{b},k} \geq N_{\mathrm{w},k}$.
If $N_{\mathrm{b},k} = N_{\mathrm{w},k}$, we also suppose w.l.o.g.\ that $N_{\mathrm{b},k-1} \geq N_{\mathrm{w},k-1}$.

First we show by contradiction that always $N_{\mathrm{b},k-1} \geq N_{\mathrm{w},k-1}$. Otherwise if $N_{\mathrm{b},k-1} < N_{\mathrm{w},k-1}$, then $N_{\mathrm{b},k} > N_{\mathrm{w},k}$
(note that $N_{\mathrm{b},k} = N_{\mathrm{w},k}$ would imply $N_{\mathrm{b},k-1} \geq N_{\mathrm{w},k-1}$).
This can happen only when a black item goes into a white bin, but then the numbers of black and white bins are swapped, hence black
and white were already the two strictly most frequent colors before the $k$-th item
which contradicts the assumption of the claim.
We conclude that $N_{\mathrm{b},k} \geq N_{\mathrm{w},k}$ and $N_{\mathrm{b},k-1} \geq N_{\mathrm{w},k-1}$.

Before the $k$-th item the number of non-black bins is at most
$\lceil 1.5D_{k-1}\rceil - N_{\mathrm{b},k-1} = D_{k-1} - \alpha_{\mathrm{b},k-1}$,
since there are at most $\lceil 1.5D_{k-1}\rceil$ bins by Claim~\ref{cl:BAFineqForColors}
(we use that the main invariant holds before packing the $k$-th item).
As we have $N_{\mathrm{b},k-1} \geq N_{\mathrm{w},k-1}$ and
black and white are not the two strictly most frequent colors before the $k$-th item,
there must be a color r $\not\in \{ \mathrm{b}, \mathrm{w} \}$ such that $N_{\mathrm{r},k-1} \geq N_{\mathrm{w},k-1}$
(let the color be red w.l.o.g.). Therefore the number of white bins is at most half of the number of non-black bins,
i.e., $N_{\mathrm{w},k-1} \leq (D_{k-1} - \alpha_{\mathrm{b},k-1})/2$.

We show by contradiction that the $k$-th item must be packed using the second case of the algorithm.
(Note that BAF cannot use the first case, since otherwise all bins would have the same color after
packing the $k$-th item.)
If the item is packed using the third case, it must hold that
$N_{\mathrm{b},k-1} \geq \lceil D_{k-1}/2\rceil +1$ and $N_{\mathrm{r},k-1}\geq \lceil D_{k-1}/2\rceil +1$.
Since there are at most $\lceil 1.5D_{k-1}\rceil$ bins by Claim~\ref{cl:BAFineqForColors},
we get $N_{\mathrm{w},k-1} \leq \lfloor D_{k-1}/2\rfloor -2$, but then the $k$-th item
cannot cause $N_{\mathrm{w},k} > N_{\mathrm{r},k}$.

Therefore BAF packs the $k$-th item using the second case
and it follows that the main invariant holds after packing the $k$-th item for all colors
by Claim~\ref{clm:OptAlgMainInvar}.

Observe that by packing the $k$-th item, the number of white bins must increase,
or the number of red bins must decrease, or both.
Note that the $k$-th item can have any color, not only white. 
We distinguish two cases: The $k$-th item is white
and the $k$-th item is not white.

If the $k$-th item is white,
we have $\alpha_{\mathrm{b},k} \leq \alpha_{\mathrm{b},k-1}$, as the number of black bins does not increase
(note that there is an inequality because of a possible increase of $D$ or a decrease of $N_{\mathrm{b}}$).
We get

\begin{align*}
\alpha_{\mathrm{w},k} &= N_{\mathrm{w},k} - \left\lceil\frac{D_{k}}{2}\right\rceil
= N_{\mathrm{w},k-1} + 1 - \left\lceil\frac{D_{k}}{2}\right\rceil
\leq \frac{D_{k-1} - \alpha_{\mathrm{b},k-1}}{2} + 1 - \left\lceil\frac{D_{k}}{2}\right\rceil \\
&\leq \frac{D_{k} - \alpha_{\mathrm{b},k}}{2} + 1 - \left\lceil\frac{D_{k}}{2}\right\rceil
\leq -\frac{\alpha_{\mathrm{b},k}}{2} + 1.
\end{align*}
where we used $N_{\mathrm{w},k-1} \leq (D_{k-1} - \alpha_{\mathrm{b},k-1})/2$ for the first inequality and
$D_{k-1} - \alpha_{\mathrm{b},k-1} \leq D_{k} - \alpha_{\mathrm{b},k}$ for the second inequality
which follows from $\alpha_{\mathrm{b},k} \leq \alpha_{\mathrm{b},k-1}$.

We know that
$\alpha_{\mathrm{w},k}\leq -\alpha_{\mathrm{b},k}/2 + 1$.
Therefore
$$2\alpha_{\mathrm{w},k}+2\alpha_{\mathrm{b},k}
\leq -\alpha_{\mathrm{b},k} + 2 + 2\alpha_{\mathrm{b},k}
= \alpha_{\mathrm{b},k} + 2 \leq \vari{CD}_{\mathrm{b},k} + 2
\leq \vari{CD}_{\mathrm{w},k} + \vari{CD}_{\mathrm{b},k} + 1
$$
where we use the main invariant~(\ref{eq:mainInvariant}) for black color for the second inequality
and $\vari{CD}_{\mathrm{w},k} \geq 1$ for the third inequality which holds, because the $k$-th item is white.

Otherwise the $k$-th item is not white and it is packed into a bin of another color than black and white,
otherwise $N_{\mathrm{b}}$ or $N_{\mathrm{w}}$ decreases, thus black
and white cannot become the two strictly most frequent colors. After packing the $k$-th item
we have $\alpha_{\mathrm{b},k} \leq \alpha_{\mathrm{b},k-1} + 1$, as the $k$-th item may be black,
therefore $D_{k-1} - \alpha_{\mathrm{b},k-1} \leq D_{k} - \alpha_{\mathrm{b},k} + 1$.
Since the number of white bins does not change, we get
\begin{align*}
\alpha_{\mathrm{w},k} &= N_{\mathrm{w},k} - \left\lceil\frac{D_{k}}{2}\right\rceil
= N_{\mathrm{w},k-1} - \left\lceil\frac{D_{k}}{2}\right\rceil
\leq \frac{D_{k-1} - \alpha_{\mathrm{b},k-1}}{2} - \left\lceil\frac{D_{k}}{2}\right\rceil \\
&\leq \frac{D_{k} - \alpha_{\mathrm{b},k} + 1}{2} - \left\lceil\frac{D_{k}}{2}\right\rceil
\leq -\frac{\alpha_{\mathrm{b},k}}{2} + 0.5.
\end{align*}

In this case we have $\alpha_{\mathrm{w},k}\leq -\alpha_{\mathrm{b},k}/2 + 0.5$.
Therefore
$$2\alpha_{\mathrm{w},k}+2\alpha_{\mathrm{b},k}
\leq -\alpha_{\mathrm{b},k} + 1 + 2\alpha_{\mathrm{b},k}
= \alpha_{\mathrm{b},k} + 1 \leq \vari{CD}_{\mathrm{b},k} + 1
\leq \vari{CD}_{\mathrm{w},k} + \vari{CD}_{\mathrm{b},k} + 1
$$
where we use the main invariant~(\ref{eq:mainInvariant}) for black color for the second inequality.
Hence the secondary invariant~(\ref{eq:secondaryInvariant}) holds.
\end{proof}

\begin{figure}[!ht]
\centerline{\includegraphics[height=1.9cm]{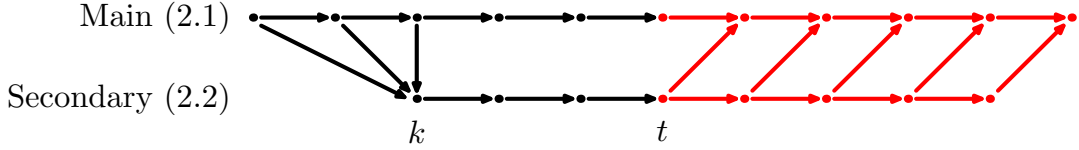}}
\caption{An illustration of dependencies of the main and secondary invariants.
The horizontal axis represents time. 
An invariant at a certain time represented by point $P$ follows from invariants from which there is an arrow to $P$.
After packing the $k$-th item (time $k$) black and white become the two strictly most frequent and
after the $t$-th item (time $t$) it starts to hold that $N_{\mathrm{b}} > \lceil D/2\rceil$ and
$N_{\mathrm{w}} > \lceil D/2\rceil$.
Thus in the black part of the figure, BAF uses the first or the second case of the algorithm,
while in the red part BAF uses the third case of the algorithm.}
\label{fig:OptAlgInvariants}
\end{figure}

We now complete the proof of the theorem by putting everything together.
Precisely, we prove that the main invariant holds during
the whole run of the algorithm by induction. 
The main invariant for each color holds trivially at the beginning before any item comes.
When the $t$-th item is packed, there are two cases:
\begin{compactitem}
\item No two colors were the strictly most frequent before the $t$-th item:
BAF keeps the main invariant for all colors by Claim~\ref{clm:OptAlgMainInvar}, since it
must pack the $t$-th item with the first or the second case of the algorithm.
If two colors become the two strictly most frequent after packing the $t$-th item,
the secondary invariant starts to hold by Claim~\ref{clm:OptAlgBaseCaseSecondary};
otherwise the secondary invariant is irrelevant in this case.
\item Two colors were the strictly most frequent:
Let these two colors be black and white w.l.o.g.
Then the main invariant for all colors and the secondary invariant for black and white are kept
by Claim~\ref{clm:OptAlgInductionSecondary}
(even if black and white are not the two strictly most frequent after the $t$-th item).

It may happen that the two strictly most frequent colors change after packing the $t$-th item
(e.g., to black and red). The main invariant for all colors still follows by Claim~\ref{clm:OptAlgInductionSecondary},
but the secondary invariant for the new strictly most frequent colors follows by Claim~\ref{clm:OptAlgBaseCaseSecondary}.
\end{compactitem}

See Figure~\ref{fig:OptAlgInvariants} for an illustration of dependencies of the invariants.

Therefore we can keep the main invariant
$N_c - \lceil D/2\rceil \leq \vari{CD}_c$ for all colors $c$ during
the whole run of the algorithm and the theorem follows by
Claim~\ref{cl:BAFineqForColors}.
\end{proof}

\section{Algorithms for Items of Arbitrary Size}\label{sec:anySize}

\subsection{Lower Bound on Competitiveness of Any Online Algorithm}\label{sec:anySizeLB}

We use the construction by D\'osa and Epstein~\cite{DosEps14} proving the lower bound 2 for two colors
to get a lower bound $2.5$ using three colors. We combine it with the hard instance that shows the lower bound
$1.5$ for zero-size items. 

\begin{theorem} \label{thm:LBforMoreColorsAndGeneralSizes}
For items of at least three colors, there is no deterministic online algorithm with an asymptotic
competitive ratio less than $2.5$. 
\end{theorem}

\begin{proof}
Throughout the whole proof the adversary uses only three colors denoted by black, white and red 
and abbreviated by b, w and r in formulas. 
Let $n > 1$ be a large integer.
The adversary starts with the hard instance for zero-size items from the proof of Theorem~\ref{thm:LBforColors}
with the optimum equal to $n$. By Lemma~\ref{l:LBforColorsInstance} there are at least $\lceil 1.5n \rceil$ bins 
of the same color, w.l.o.g.\ white.
Let $W$ be the set of bins that are white after the first part of the instance ($|W| \geq \lceil 1.5n \rceil$).
We also know that $\vari{CD}_{\mathrm{w}} \leq n$, $\vari{CD}_{\mathrm{b}} = 0$,
and $\vari{CD}_{\mathrm{r}} = 0$. 

The second part of the instance is a slightly simplified construction by D\'osa and Epstein~\cite{DosEps14}.
Their idea goes as follows: The adversary sends the instance in phases, 
each starting with white and black small items. If the black item is put into an already opened bin,
we send a huge white item that can be put only on the small black item.
Therefore the algorithm has to put the huge white item in a new bin, but an optimal offline
algorithm puts the small black item into a new bin and the huge white item on it.
Otherwise, if the small black item is put into a new bin, the phase is finished:
The online algorithm opened a bin in the phase, while an optimal offline algorithm does not need to.
This way an online algorithm is forced to behave oppositely to an optimal offline algorithm.

We formalize this idea by the following adversarial algorithm. Let $\varepsilon = 1/(6n)$ and
$\delta_i = 1/(5^i \cdot 6n)$. The adversary uses the items of the following types:
\begin{compactitem}
\item regular white items of size $\varepsilon$,
\item regular black items of size $\delta_i$ for some $i \geq 1$,
\item special black items of size $3\delta_i$ for some $i \geq 1$,
\item huge white items of size $1 - 2\delta_i$ for some $i \geq 1$.
\end{compactitem}

Let $i$ be the index of the current phase and let $j$ be the number of huge white items
in the instance so far. The adversarial algorithm is as follows:
\begin{compactenum}
\item Let $i = 0$ and $j = 0$.
\item If $j = n$ or if $i = \lceil 2.5n\rceil$, then stop.
\item Let $i = i + 1$. Send $\left(\nolinefrac{\mathrm{white}}{\varepsilon}, \nolinefrac{\mathrm{black}}{\delta_i}\right)$, i.e., a group consisting
of a regular white item and a regular black item.
\item If the regular black item goes to a new bin or to a bin with level zero, go to the step 2
(continue with the next phase).
\item Let $j = j + 1$. Send
$\left(\nolinefrac{\mathrm{black}}{3\delta_i}, \nolinefrac{\mathrm{white}}{1 - 2\delta_i}, \nolinefrac{\mathrm{black}}{\delta_i}\right)$.
Then go to the step 2 (continue with the next phase).
\end{compactenum}

First we show that we can pack the whole list of items into $n + 1$ bins and then that no huge white item
can be packed by an online algorithm into a bin from the set $W$, i.e., one of $\lceil 1.5n \rceil$ bins which are white after the first part
with zero-size items.

\begin{lemma}
$\vari{OPT} = n + 1.$
\end{lemma}

\begin{proof}
The first part of the items, i.e., the hard instance for zero-size items, has the optimum exactly $n$.
Moreover, all bins are white by Lemma~\ref{l:LBforColorsInstance}.
Each of $j \leq n$ huge white items is packed with the two regular black items from the same phase,
thus creating $j$ full bins with a black item at the bottom.
All these bins thus can be combined with the bins from the first part of the instance.
The remaining items have alternating colors and total size of at most
$i \cdot 2\varepsilon \leq \lceil 2.5n\rceil \cdot 2/(6n) \leq 1$
(recall that $i$ is the index of the current phase and that each black item is smaller than $\varepsilon$),
so they can be put into an additional bin.

Note that the maximal discrepancy $\vari{LB}_2$ is $n + 1$, since we end the first part by $n$ white items and 
start the second part by a regular white item. Hence $\vari{OPT} = n + 1$.
\end{proof}

We now analyze how an online algorithm behaves on the second part of the instance.

\begin{lemma}
After the $i$-th phase the number of bins with a non-zero level is at least~$i$.
Moreover, no huge white item goes into a bin from the set $W$.
\end{lemma}

\begin{proof}
We use the fact that $3\delta_i < \varepsilon$, i.e., all black items are smaller than $\varepsilon$,
and that a huge white item of size $1 - 2\delta_i$ cannot be packed with a black item of size at least $\delta_j$ for any $j < i$.

We show that in each phase the number of bins with a non-zero level increases by at least one.
This holds trivially, if the second item in a phase, denoted by $s$, goes to a new bin or to a bin with level zero.
Otherwise, if $s$ is put into a bin of non-zero level, the adversary continues the phase by sending
three other items, most importantly a huge white item~$h$. The item~$s$ is the only one sent before $h$
that is sufficiently small to be packed into a single bin with $h$, but $s$ is in a bin with another non-zero item.
Therefore $h$ must be packed into a new bin or a bin with level zero. This proves the first statement of the lemma.

For the second statement, note that if $h$ goes to a bin with zero-size items only, the bin cannot be white,
but all the bins from the set $W$ that have still level zero (while packing $h$) are white.
As we already observed, $h$ also does not go to a bin from $W$ that has a non-zero level.
\end{proof}

By the previous lemma we know that if the second part of the instance ends with $i = \lceil 2.5n\rceil$,
there are $\lceil 2.5n\rceil$ non-zero bins. If the instance stops by $j = n$,
the online algorithm has at least $|W| + n \geq \lceil 2.5n\rceil$ open bins,
since it opens bins in $W$ on the first part of the instance and it must put $n$ huge white items into other bins.
As $\vari{OPT} = n + 1$, we get that the ratio between the online algorithm
and the optimum tends to $2.5$ as $n$ goes to infinity.
\end{proof}

Note that if we replace the first part of the instance by $n$ white items, we get the lower bound of 2 using only two colors.

\subsection{$\mathbf{3.5}$-competitive Algorithm}\label{sec:anySizeUB}

We now show that there is a constant competitive online algorithm even for items of sizes between 0 and 1.
We combine algorithms Pseudo from~\cite{balogh14} and our algorithm BAF that is $1.5$-competitive for zero-size items.
The algorithm Pseudo uses \textit{pseudo bins} which are bins of unbounded capacity.

\begin{center}
\fbox{\parbox{0.9\textwidth}{
{\bf Pseudo-BAF:}
\begin{compactenum}
\item First pack an incoming item into a pseudo bin using the algorithm BAF (treat the item as a zero-size item).
\item In each pseudo bin, items are packed into unit capacity bins using Next Fit.
\end{compactenum}
}}
\end{center}

\begin{theorem}
\label{thm:PseudoBAF}
The algorithm Pseudo-BAF for Colored Bin Packing is absolutely
$3.5$-com\-pe\-ti\-ti\-ve. In the parametric case when items have size
at most $1/d$, for a~real $d \geq 2$, it uses at most $\lceil (1.5 +
d/(d-1))\vari{OPT}\rceil$ bins. Moreover, the analysis is
asymptotically tight.
\end{theorem}

\begin{proof}
In the general case for items between 0 and 1 we know that two
consecutive bins in one pseudo bin have total size strictly more than one,
since no two consecutive items of the same color are in a pseudo bin.
In each pseudo bin we match each bin with an odd index with the following bin
with an even index, therefore we match all bins except at most one in each pseudo bin.
Moreover, the total size of a pair of matched bins is more than one.
Therefore the number of matched bins is strictly less than $2 \cdot \vari{LB_1} \leq 2 \cdot \vari{OPT}$,
i.e., at most $2 \cdot \vari{OPT} - 1$.
The number of not matched bins is at most the number of pseudo bins
created by the algorithm BAF which uses at most
$\lceil 1.5 \cdot \vari{LB_2}\rceil \leq \lceil 1.5 \cdot \vari{OPT}\rceil \leq 1.5 \cdot \vari{OPT} + 0.5$ bins.
Summing both bounds, the algorithm Pseudo-BAF creates at most $3.5 \cdot \vari{OPT}$ bins.

For the parametric case, inside each pseudo bin all real bins except the last one
have level strictly more than $(d - 1)/d$, so their number is strictly less than $d/(d-1) \cdot \vari{OPT}$,
i.e., at most $\lceil d/(d-1) \cdot \vari{OPT}\rceil - 1$.  
The number of pseudo bins is still bounded by $\lceil 1.5 \cdot \vari{OPT}\rceil$,
thus the algorithm Pseudo opens at most $\lceil (1.5 + d/(d-1))\vari{OPT}\rceil $ bins.

We show the tightness of the analysis by combining hard instances
for Pseudo by Balogh et al.~\cite{balogh14}
and for BAF from the proof of Theorem~\ref{thm:LBforColors}.
More concretely, for $n$ (a big integer) let $\varepsilon = 1/(2n)$.
The adversary sends $n - 1$ groups of three items, specifically
$(n - 1) \times \left( \nolinefrac{\mathrm{white}}{\varepsilon}, \nolinefrac{\mathrm{black}}{1},
\nolinefrac{\mathrm{black}}{\varepsilon} \right).$

The algorithm creates one pseudo bin containing every first and second item from each group and $n - 1$ pseudo bins,
each containing only the third item from a group.
Moreover, the first pseudo bin is split into $2 \cdot (n - 1)$ unit capacity bins (each item is in a separate bin), so there are $3 \cdot (n - 1)$ bins.
The optimum for $n - 1$ groups is $n$, because we can pack all tiny items together in one bin and $\vari{LB_1} = n$.

Then the adversary sends the hard instance with zero-size items from the proof of Theorem~\ref{thm:LBforColors}
and BAF creates additional $\lceil (n - 1)/2 \rceil$ pseudo bins,
while the optimum on the instance is $n - 1$. Pseudo-BAF now have $\lceil 3.5 \cdot (n - 1) \rceil$ bins.
Observe that the optimal packing for $n - 1$ groups does not need to be changed to put there zero-size items, thus $\vari{OPT} = n$.

For the parametric case, the adversary uses a modification of the first part of the hard
instance by Balogh et al.~\cite{balogh14} on which Pseudo creates at least $(d - 1) n + dn$ bins,
while its optimal packing needs $(d - 1) n + 1$ bins.
Moreover, the adversary can continue with the hard instance with zero-size items
like in the general case and force Pseudo-BAF to create additional $\lceil (d - 1) n/2\rceil$ bins
without increasing the optimum. Therefore
Pseudo-BAF ends up with asymptotically $(1.5 + d/(d-1)) \vari{OPT}$ bins.
\end{proof}

\subsection{Classical Any Fit Algorithms} \label{sec:AFalgsOnColors}
We analyze algorithms First Fit, Best Fit and Worst Fit and we find that
they are not constant competitive.  Their competitiveness cannot be
bounded by any function of the number of colors even for only three
colors, in contrast to their good performance for two colors.

\begin{proposition}
First Fit and Best Fit are not constant competitive.
\end{proposition}

\begin{proof}
The adversary sends an instance of $4n$ items which can be packed into two bins,
but FF and BF create $n + 1$ bins where $n$ is an arbitrary integer.

Let $\varepsilon = 1/(4n)$. The instance is $n \times \left( \nolinefrac{\mathrm{black}}{\varepsilon}, \nolinefrac{\mathrm{black}}{\varepsilon},
\nolinefrac{\mathrm{white}}{\varepsilon}, \nolinefrac{\mathrm{red}}{\varepsilon} \right)$.
An optimal packing can be obtained by putting black items from each group into the first and the second bin,
the white item into the first bin and the red item into the second bin.

FF and BF pack the first group into two bins, both with a black bottom item,
and white and red items fall in the first bin.
The first black item, the white item and the red item from each following group are packed into the first bin,
while the second black item is packed into a new bin.
Therefore these algorithms create one bin with all white and red items and all first black items from each group and $n$ bins with a single black item.

Hence FF and BF create $(n+1)/2\cdot \vari{OPT}$ bins.
\end{proof}

Note that WF on such instance creates an optimal packing,
but the instance can be modified straightforwardly to obtain a bad behavior for WF.

\begin{proposition}
Worst Fit is not constant competitive.
\end{proposition}

\begin{proof}
The instance is similar to the one in the previous proof, but sizes of items are different in each group.
Let $\varepsilon = 1/(2n)$ and let $\delta = 1/(6n^2 + 1)$. The instance is
$n \times \left( \nolinefrac{\mathrm{black}}{\delta}, \nolinefrac{\mathrm{black}}{\varepsilon},
\nolinefrac{\mathrm{white}}{\delta}, \nolinefrac{\mathrm{red}}{\delta} \right)$.
 
We observe that the optimal packing does not change with other sizes.
However, WF packs all $\delta$-items into the first bin, i.e.,
first black items from each group and all white and red items,
since the level of the first bin stays at most $(3n)/(6n^2 + 1)$,
which is less than $1/(2n)$ as $\varepsilon/\delta > 3n$.
Therefore all second black items are packed into separate bins
and WF creates $n + 1$ bins, while the optimum is two.
\end{proof}

\section{Black and White Bin Packing}\label{sec:BWBP}

For sequences with only two colors,
we improve the upper bound on the absolute competitive ratio of Any Fit algorithms from 5 to 3.
Then we show that Worst Fit performs even better
for items with size of at most $1/d$ (for $d \geq 2$)
as it is $(1 + d/(d-1))$-competitive in this case.
Both bounds are tight by the results of Balogh et al.~\cite{balogh14}
(more precisely, the bound for WF is tight only for an integer $d \geq 2$),
therefore it matches the performance of Pseudo, the online algorithm with the best competitive
ratio known so far. Note that for infinitesimally small items WF is $2$-competitive,
while BF and FF remain $3$-competitive.

Both bounds are tight, since there are instances for FF and
BF on which the competitive ratio is asymptotically 3 and an instance
for WF in the parametric case on which WF uses asymptotically $(1 +
d/(d-1))\vari{OPT}$ bins for an integer $d \geq 2$~\cite{balogh14}
(more precisely, the bound for WF is tight only for an integer $d \geq 2$).

\subsection{Competitiveness of Any Fit Algorithms} \label{sec:AF-BWBP}
\begin{theorem}
\label{thm:AFcomp}
Any algorithm in the Any Fit family is absolutely $3$-competitive for
Black and White Bin Packing.
\end{theorem}

\begin{proof}
We use the following notation: An item is \textit{small}
when its size is less than $0.5$ and \textit{big} otherwise.
Similarly \textit{small bins} have level less than $0.5$ and \textit{big bins} have level at least $0.5$.

We assign bins into \textit{chains} --- sequences of bins in which all bins except the last must be big.
If there is only one bin in a chain it must be big.
Moreover, the bottom item in the $i$-th bin of a chain cannot be added into the $(i- 1)$-st bin ---
even if it would have the right color, i.e., it is too big to be put into the $(i- 1)$-st bin.

A bin is contained in at most one chain. We call a bin that is not in a chain a \textit{separated} bin.
We create chains such that all big bins are in a chain and only as few small bins as possible remains separated.
Moreover, our chains can have at most two bins, so the average level
of bins in each chain is clearly at least $0.5$.

It follows that the total number of bins in all chains
is bounded from above by $2 \cdot \vari{OPT}$.
We want to bound the number of separated bins from above by the maximal color discrepancy $\vari{LB_2}$
which yields the $3$-competitiveness of AF.

We define a process of assigning bins into chains. 
We simply try to put as many bins into chains as possible,
but we add a bin into a chain only when the last bin
in the chain has another color than the bottom item of the added bin.

Formally, when an item from the input sequence is added we do the following:
\begin{compactitem}
\item The item is added into a bin in a chain: Nothing happens with chains or separated bins.
\item The item is added into a small separated bin: If the bin becomes big we create a chain from the bin,
otherwise the bin stays separated.
\item The item is big and creates a new bin: The newly created bin forms a new chain.
\item The item is small and creates a new bin:
If there is a chain in which the last bin has an item of another color on the top, i.e., black for a white incoming item and white for a black incoming item,
we add the newly created bin into the chain. (Note that the last bin in the chain must be big,
otherwise the item would fit into the last bin.)
If there is no such chain, the new bin is separated.
\end{compactitem}

Moreover, whenever a chain has two big bins
we split it into two chains, each containing one big bin. 
Therefore each chain is either one big bin, or a big bin and a small bin.
The intuitive reason for splitting chains is that
we can put more newly created small bins into chains.

If there is no separated bin at the end (after the last item is added),
we have created at most $2 \cdot \vari{OPT}$ bins.
Otherwise we define $k$ and $t$ as indexes of incoming items
and show that the color discrepancy of items between the $k$-th and the $t$-th item is at least the number of separated bins at the end.

Let $t$ be the index of an item that created the last bin that is separated when it is created (the $t$-th item must be small).
Suppose w.l.o.g.\ that the $t$-th item is black.
Note that a small item that comes after the $t$-th item can create a bin,
but we put the bin into a chain immediately,
therefore the number of separated bins can only decrease after adding the $t$-th item.

Let $b_i$ be the number of small black bins, i.e., bins with a black item on the top,
and $w_i$ be the number of small white bins
after adding the $i$-th item from the sequence. From the definition of $t$ we know that $w_t = 0$.

We define $k$ as the biggest $i \leq t$ such that $b_i = 0$, i.e., there is no small black bin
(if $b_i > 0$ for all $i \geq 1$ we set $k = 0$).
Clearly the $(k + 1)$-st item must be small and black.
Note that there can be some separated white bins and possibly some other small white bins in chains,
but there is no separated black bin.
Let $W$ be the set of white bins that are separated after adding the $k$-th item.
Before adding the $t$-th item and creating the last bin, all bins in $W$ must have a black item on the top, or become big bins in chains (thus $k \leq t - |W|$).

Let \textit{new items} be items with an index $i$ such that $k < i \leq t$.
We want to bound the number of separated bins after adding the $t$-th item by the color discrepancy.
Note that these bins are small by the process of assigning bins into chains.
We observe that all separated bins must have a black item on the top before adding the $t$-th item
and also all chains have a black item on the top in the last bin,
otherwise the bin created by the $t$-th item would be added in a chain.

Hence for separated bins with a black item at the bottom the number of black items is greater by one than the number of white items.
Separated bins created with a new item must have a black item at the bottom,
since otherwise there cannot be a small black bin and $b_i = 0$ for $k < i < t$.

Separated bins from the set $W$ can have the same number of black and white items before adding the $t$-th item,
but in each such bin there is one more new black item than new white items,
since the first and the last such items are black.

Now we look at new items which are packed into bins that are in chains after adding the $t$-th item.
We call such an item a \textit{link}.
Note that some links can be at first packed into separated bins, but these bins are put into chains before adding the $t$-th item.
It suffices to show the following claim. 

\begin{claim}
In each chain the number of black links is at least the number of white links after adding the $t$-th item.
\end{claim}

\begin{proof}

When the $t$-th item comes and creates a new separated bin, the last item in each chain must be black.
Therefore the claim holds for the chains with only one bin.

For the chains with two bins (the first big and the second small) we observe 
that a bin created with a link has either a black item, or a big white item at the bottom.
If it would have a small white link at the bottom, there cannot be a small black bin and $b_i = 0$ for $k < i < t$
which is a contradiction with the definition of $k$.
Since a big white item starts a new chain, the second bin in a chain cannot have a white link at the bottom.

Moreover, the first link in the second bin of a chain must be black,
because either the second bin was created after the $k$-th item and we use the observation from the previous paragraph,
or it was created before the $k$-th item and then it must had a white item on the top when the $k$-th item came,
since there was no small bin with a black item on the top.

So it cannot happen in a chain with two bins that there are two white links next to each other,
or separated by some items that are not links.
Note that for black links this situation can happen.
Since there must be a link in the second bin and the last such link is black,
the claim holds for chains with one big and one small bin.

The process of assigning bins into chains does not allow chains with more than two bins or with two big bins.
Hence in each chain the number of black links is at least the number of white links.
\end{proof}

Let $s$ be the number of separated bins.
We found out that when we focus on new items, i.e., items with an index $i$ such that $k < i \leq t$,
there is one more such black item than such white items in all separated bins
and at least the same number of such items of both colors in bins in all chains, i.e., links.
Moreover, after the $t$-th item comes $s$ can only decrease, since no separated bin is created.
So we have bounded the value of $s$ at the end from above by the color discrepancy between the $(k + 1)$-st and the $t$-th item:

$$s \leq \left| \sum_{\ell=k + 1}^{t} s_\ell \right| \leq \vari{LB_2}$$

where $s_i$ is $1$ when the $i$-th item is white and $-1$ otherwise. 
Note that some items after the $t$-th item can create a bin,
but such bins are put into chains.
\end{proof}

\subsection{Competitiveness of the Worst Fit Algorithm} \label{sec:WF-BWBP}

The Worst Fit algorithm performs in fact even better when all items are small
which we prove similarly to the proof of Theorem~\ref{thm:AFcomp}.

\begin{theorem}
\label{thm:WFparam}
Suppose that all items in the input sequence have size of at most $1/d$,
for a real $d \geq 2$.  Then Worst Fit is absolutely $(1 +
d/(d-1))$-competitive for Black and White Bin Packing.
\end{theorem}	

\begin{proof}
Let $\vari{OPT}$ be the number of bins used in an optimal packing.
We divide bins created by WF into sets $B$ (big bins) and $S$ (small bins).
Each \textit{big bin} has level at least $(d - 1)/d$, thus $|B| \leq d/(d-1) \cdot \vari{OPT}$.
\textit{Small bins} are smaller than $(d - 1)/d$, thus they can receive any item of the right color.
We show that $|S|$ is bounded by the maximal color discrepancy $\vari{LB_2}$ and we obtain that WF is $(1 + d/(d-1))$-competitive.

As items are arriving, we count the number of small black bins, i.e.,
bins with a black item on the top and with level less than $(d - 1)/d$.
Let $b_i$ be the number of small black bins after adding the $i$-th item from the sequence.
Similarly let $w_i$ be the number of white bins with level less than $(d - 1)/d$ after adding the $i$-th item. 

If $b_n = 0$ and $w_n = 0$, i.e., there is no small bin at the end, WF created at most $d/(d-1) \cdot \vari{OPT}$ bins.
Otherwise suppose w.l.o.g.\ that the last created bin has a black item at the bottom.
Let $t$ be the index of the black item that created the last bin.
It holds that $w_t = 0$, since otherwise the $t$-th item would go into a small white bin.

Let $k$ be the last index smaller than $t$ for which $b_k = 0$ (if $b_i > 0$ for all $i\in\{1, \dots, t \}$, we set $k = 0$).
The $(k + 1)$-st item must be black.
We observe that any bin created after this point has a black item at the bottom,
otherwise $b_i = 0$ for some $i$ such that $k < i < t$.
Note that $w_k$ can be greater than $0$, i.e., there can be some small white bins
and the $(k + 1)$-st item goes into one of them.
Let $W$ be the set of these bins.
Before adding the $t$-th item and creating the last bin, all bins in $W$ must have a black item on the top, or become big bins (thus $k \leq t - |W|$).

Let \textit{new items} be items with an index $i$ such that $k < i \leq t$.
We want to bound the number of small bins after adding the $t$-th item by the color discrepancy.
We already observed that all these bins must have a black item on the top.
Hence for small bins with a black item at the bottom the number of black items is greater by one than the number of white items.
Small bins from the set $W$ can have the same number of black and white items,
but in each such bin there is one more new black item than new white items, since the first such item is black.

Now we look at new items which are packed
into bins that are big after the $t$-th item comes.
It suffices to show that the number of such new black items is at least the number of such new white items.
We observe that WF packs any new white item into a small bin,
otherwise $b_\ell = 0$ for some $\ell$ such that $k < \ell < t$.
Hence any new white item must be packed into a bin created after the $k$-th item (therefore with a new black item at the bottom),
or into a bin from the set $W$.
Since the first item that falls into a bin from $W$ after the $k$-th item is black, our claim holds.

Note that this matching of black and white items in big bins would fail
for algorithms like Best Fit or First Fit, since they can put
a white item into a big bin created before the $k$-th item and not contained in $W$.

We found out that when we focus on new items, i.e., items with an index $i$ such that $k < i \leq t$,
there is one more such black item than such white items in all small bins
and at least as many such black items as such white items in all big bins.
Moreover, after the $t$-th item comes the number of small bins $|S|$ can only decrease,
since no bin is created.
So we bound $|S|$ from above by the color discrepancy
between the $(k + 1)$-st and the $t$-th item:
$$|S| \leq \left| \sum_{\ell=k + 1}^{t} s_\ell \right| \leq \vari{LB_2}$$
where $s_i$ is $1$ when the $i$-th item is white and $-1$ otherwise.

Note that the last bin is already counted in the color discrepancy,
since its bottom item is black and has index $t$.
\end{proof}

\section*{Conclusions and Open Problems}

The Colored Bin Packing for zero-size items is completely solved.

For items of arbitrary size, our online algorithm still leaves
a gap between our lower bound $2.5$ and our upper bound of $3.5$.
The upper bounds are only
$0.5$ higher than for two colors (Black and White Bin Packing) where a
gap between $2$ and $3$ remains for general items.

Classical algorithms FF, BF and WF, although they maintain a constant
approximation for two colors, start to behave badly when we introduce
the third color.  For two colors, we now know their exact behavior.
In fact, all Any Fit algorithms are absolutely
$3$-competitive which is a tight bound for FF, BF and WF. However, for
items of size up to $1/d$, $d \geq 2$, FF and BF remain
$3$-competitive, while WF has the absolute competitive ratio $1 + d/(d-1)$.
Thus we now know that
even the simple Worst Fit algorithm matches the
performance of Pseudo, the online algorithm with the best competitive
ratio known so far.
It is also an interesting question whether it
holds that Any Fit algorithms cannot be better than $3$-competitive
for two colors.

{\small

\bibliographystyle{abbrv}
\bibliography{BWBP}
}

\end{document}